%% file: main.tex
\documentclass[11pt]{article}
\usepackage{amsfonts}
\usepackage{amssymb}
\usepackage{amstext}
\usepackage{amsmath}
\usepackage{xspace}
\usepackage{theorem}
\usepackage{thmtools}
\usepackage{mathtools}
\usepackage{graphicx}
\usepackage{graphics}
\usepackage{colordvi}
\usepackage{color}
\usepackage{paralist}
\usepackage{bm}
\usepackage{mathpazo}
\usepackage{boxedminipage}
\usepackage[percent]{overpic}

\advance \oddsidemargin by -0.8in

\newtheorem{theorem}{Theorem}
\newtheorem{lemma}[theorem]{Lemma}
\newtheorem{claim}[theorem]{Claim}

\newtheorem{corollary}[theorem]{Corollary}

\newtheorem{definition}{Definition}

\newenvironment{proof}{\par \smallskip{\bf Proof:}}{\hfill\stopproof}
\def\stopproof{\square}
\def\square{\vbox{\hrule height.2pt\hbox{\vrule width.2pt height5pt \kern5pt
            \vrule width.2pt} \hrule height.2pt}}

\theoremstyle{definition}

\theoremstyle{remark}

\newcommand{\OPT}{\mathrm{OPT}}
\newcommand{\capMMTC}{\mathrm{CapMMTC}}
\newcommand{\MMTC}{\mathrm{MMTC}}
\newcommand{\rMMTC}{\mathrm{RMMTC}}
\newcommand{\cov}{\mathrm{cov}}
\newcommand{\caprMMTC}{\mathrm{CapRMMTC}}
\newcommand{\opt}{T^{\star}}

\newcommand{\Real}{\mathbb{R}}
\newcommand{\bigO}{{\cal O}}

\newcommand{\calT}{\mathcal{T}}

\newcommand{\calN}{\mathcal{N}}
\newcommand{\calA}{\mathcal{A}}
\newcommand{\calM}{\mathcal{M}}
\newcommand{\calB}{\mathcal{B}}
\newcommand{\calC}{\mathcal{C}}

\newcommand{\hatT}{\widehat{T}}

\newcommand{\nullset}{\phi}

\renewcommand{\sp}{\hspace*{0.1in}}

\renewcommand{\setminus}{\backslash}

\newcommand{\restr}[1]{\ensuremath{\bigl.#1\bigr|}}

\DeclareMathOperator*{\argmax}{argmax}

\DeclarePairedDelimiter\ceil{\lceil}{\rceil}
\DeclarePairedDelimiter\floor{\lfloor}{\rfloor}

\everymath{\displaystyle}

\begin{document}

\title{A Constant Factor Approximation for \\ Capacitated Min-Max Tree Cover}

\author{
        Syamantak Das\thanks{
        IIIT Delhi, India,
        {\tt {\bf email:} syamantak@iiitd.ac.in}}
%%%%%%%%%%%%
\and
        Lavina Jain\thanks{
        IIIT Delhi, India,
        {\tt {\bf email:} lavina16052@iiitd.ac.in}}
\and
         Nikhil Kumar\thanks{
        IIT Delhi, India,
        {\tt {\bf email:} nikhil@cse.iitd.ac.in}}
}

\date{\today}

\maketitle

\begin{abstract}
    Given a graph $G=(V,E)$ with non-negative real edge lengths and an integer parameter $k$, the (uncapacitated) Min-Max Tree Cover problem seeks to find a set of at most $k$ trees which together span $V$ and each tree is a subgraph of $G$. The objective is to minimize the maximum length among all the trees. In this paper, we consider a capacitated generalization of the above and give the first constant factor approximation algorithm. In the capacitated version, there is a hard uniform capacity ($\lambda$) on the number of vertices a tree can cover. Our result extends to the rooted version of the problem, where we are given a set of $k$ root vertices, $R$ and each of the covering trees is required to include a distinct vertex in $R$ as the root. Prior to our work, the only result known was a $(2k-1)$-approximation algorithm for the special case when the total number of vertices in the graph is $k\lambda$~[Guttmann-Beck and Hassin, {\em J. of Algorithms, 1997}]. Our technique circumvents the difficulty of using the minimum spanning tree of the graph as a lower bound, which is standard for the uncapacitated version of the problem~[Even et al.,{\em OR Letters 2004}]~[Khani et al.,{\em Algorithmica 2010}]. Instead, we use {\em Steiner trees} that cover $\lambda$ vertices along with an iterative refinement procedure that ensures that the output trees have low cost and the vertices are well distributed among the trees.
\end{abstract}
\newpage

\input{capacitated}

\input{exact} 

\input{rooted}

\bibliographystyle{plain}
\bibliography{references}

\end{document}

%% file: capacitated.tex
\section{Introduction}
\label{sec:intro}
Covering vertices of a given graph using simpler structures, for example, paths, trees, stars and so on, have long attracted the attention of the Computer Science and Operations Research communities. This can be attributed to a variety of applications in vehicle routing, network design and related problems. One classical example is the so-called `Nurse Location Problem'~\cite{EvenGKRS04}. The goal is to place a group of nurses at different locations and finding a tour for each of them so that every patient is visited by a nurse. A similar setting arises in vehicle routing. Suppose we are given a set of vehicles, initially located at a given set of depots. The goal is to find a tour for each of these vehicles, each starting and ending at the respective depots so as to cover client demands at various locations. One of the most popular objectives is to minimize the maximum distance travelled by any vehicle, also known as the {\em makespan} of the solution. A standard reduction shows that this problem is equivalent, within an approximation factor of $2$, to finding trees in a graph such that all vertices in the graph are covered by the union of these trees. In this paper, we consider a natural generalization of the above setting. As before, we are given a set of vehicles, initially located at a given set of depots and a set of clients. Additionally one package is to be delivered to each of the clients and each vehicle can carry at most a fixed number of packages (all packages are identical). We have to find a set of tour such that each client receives a package and the objective, as before is to minimize the maximum distance travelled by any vehicle. This can be seen as a capacitated version of the Nurse Location Problem, where each nurse can visit at most a given number of patients. We formally define the problem now, which will be useful in further discussions.

\subsection{Notations and Preliminaries} \label{Preliminaries}
We set up some preliminaries, notations and definitions from literature that will be useful in the exposition of our contributions. The set of positive integers $\{1,2,,\ldots, n \}$ is denoted by $[n]$. GIven a graph $G=(V,E)$, $H=(V_{H},E_{H})$ is a subgraph of $G$ if $V_{H} \subseteq V$ and $E_{H} \subseteq E$. We use $\ell(H)$ to denote the total length of edges in the subgraph $H$. For any two  vertices $u,v$, $d(u,v)$ denotes the shortest path distance between $u,v$. We extend the definition to subsets of vertices $U,V$ - define $d(U,V)$ to be $\min_{u\in U, v \in V} d(u,v)$. The set of vertices in any graph $H$ is denoted by $V(H)$.
 \begin{definition}
    {\bf Min-Max Tree Cover Problem ($\MMTC$)} Given a graph $G(V,E)$, edge length $\ell:E\rightarrow \Real_{\geq 0}$ and a parameter $k$, one is required to output a set of $k$ trees $T_i$ for $i \in [k]$, such that each $T_{i}$ is a subgraph of $G$ and $\cup_{i=1}^k V(T_i) = V$. The objective is to minimize $max_{i=1}^{k} \ell(T_i)$. 
\end{definition}
Note that two trees in a feasible solution can share vertices as well as edges. In the capacitated version of the problem, we are also given an additional parameter $\lambda$. A feasible solution to the Capacitated Min-Max k-Tree Cover consists of a set of trees (not necessarily disjoint) $\calT = \{ T_{1},T_{2},\ldots,T_{k} \}$ along with an assignment of each vertex $v \in V$ to one of the trees containing it, such that no more than $\lambda$ vertices are assigned to any $T \in \calT$. The goal, as in the case of uncapacitated case, is to minimize the maximum length of any tree in $\calT$. Let $\cov(T)$ denote the set of vertices assigned to $T$. We think of the vertices in $\cov(T)$ as being covered by $T$. Hence, in any feasible assignment $\cup_{T \in \calT}\cov(T)=V$ and $|\cov(T)| \leq \lambda$ for all $T \in \calT$. Note that a vertex may be a part of multiple trees, but it is covered by exactly one of them. Whenever the notation $\cov(T)$ is used, there is an (implicit) underlying assignment of vertices to the trees. 
\begin{definition}
    {\bf Capacitated Min-Max Tree Cover ($\capMMTC$)} Given a graph $G(V,E)$, edge length $\ell:E\rightarrow \Real_{\geq 0}$, and two integer parameters $k,\lambda$, one is required to output a set of $k$ trees $T_{i}$ for $i\in [k]$ along with an assignment of every vertex in $G$ to a tree, such that each $T_{i}$ is a subgraph of $G$, $\cup_{i=1}^k \cov(T_i) = V$ and $|\cov(T_i)| \leq \lambda$ for $i \in [k]$. The objective is to minimize $max_{i=1}^k \ell(T_i)$.
\end{definition} 
In the rooted version of MMTC and CapMMTC, we are given a set of $k$ roots as well. The only additional constraint being that each tree in the output should contain a distinct root. We will refer to these problems as Rooted Min-Max Tree Cover Problem($\rMMTC$) and Capacitated Rooted Min-Max Tree Cover ($\caprMMTC$).

\subsection{Our Contribution}
We give the first polynomial time constant factor approximation algorithm for $\capMMTC$ and $\caprMMTC$. 
\begin{theorem} \label{thm:capmmtc}
There is a polynomial time $\bigO(1)$-approximation algorithm for $\capMMTC$. 
\end{theorem} 
\begin{theorem} \label{thm:caprmmtc}
There is a polynomial time $\bigO(1)$-approximation algorithm for $\caprMMTC$. 
\end{theorem} 
Our algorithms are much more intricate and involved than those which give an $\bigO(1)$ approximation for the uncapacitated case, as is generally the case with capaciated versions of many problems in combinatorial optimization, for instance k-median and facility location problems. These are the first approximation algorithms for both the problems, to the best of our knowledge. The only known result is a $(2k-1)$-approximation for the special case when the total number of vertices in the graph is $k\lambda$ and hence every tree must contain exactly $\lambda$ vertices~\cite{GuttmanBH97}. All our algorithms are combinatorial. We first show the result for the unrooted case, ie. $\capMMTC$ and then extend the ideas to prove the result for the rooted problem $\caprMMTC$. We show that value of the constant in Theorem \ref{thm:capmmtc} and \ref{thm:caprmmtc} is $\leq 300$. For simplicity of exposition, we have not tried to optimize the constant. We believe that it should be possible to do so with some more effort and leave it as an open problem. We prove Theorem~\ref{thm:capmmtc} in Section~\ref{sec:capacitated} and defer the proof of Theorem~\ref{thm:caprmmtc} to the Appendix, Section~\ref{sec:rooted}. 

\subsection{Related Work}
Even et al.~\cite{EvenGKRS04} and Arkin et al.~\cite{ArkinHL06} gave $4$-approximation algorithms for both rooted and unrooted (uncapacitated) $\MMTC$. Khani et al. improved the unrooted version to a 3-approximation~\cite{KhaniS14}, whereas Nagamochi and Okada~\cite{NagamochiK07} gave a $3-2/(k+1)$-approximation for the special case of the rooted version where all roots are co-located. On the other hand, the $\MMTC$ problem has been shown to be hard to approximate to a factor better than $3/2$, assuming $P\neq NP$~\cite{ZhouQ10}. The problem has been considered when the underlying graph has special structure. The rooted version of the $\MMTC$ problem on a tree admits a PTAS, as shown recently by Becker and Paul~\cite{BeckerP18}, while the unrooted version has a $(2+\epsilon)$-approximation, given by Nagamochi and Okada~\cite{NagamochiK07}. Further, Chen and Marx gave fixed parameter tractable algorithms for the problem on a tree ~\cite{ChenM18}. $\MMTC$ on a star is equivalent to the classical makespan minimization problem on identical machine, for which elegant EPTAS-es are well known~\cite{AlonAWY98}. In stark contrast, no approximation algorithms have been reported so far for the capacitated versions of these problems, to the best of our knowledge. Guttmann-Beck and Hassin gave an $(2k-1)$-approximation for the special case  where the number of vertices in each the tree is exactly $\lambda$~\cite{GuttmanBH97}. On the other special case of a star metric, an EPTAS result follows from the work on identical machine scheduling with capacity constraints~\cite{ChenJLZ16}. Interestingly, the problem becomes inapproximable if one disallows sharing of vertices and edges between the trees in the solution~\cite{GuttmanBH97}, even when $k=2$. A related but quite different setting is that of bounded capacity vehicle routing with the objective of minimizing the total length of all the tours. Here, a vehicle is allowed to make multiple tours to cover all the points, however, each tour can serve at the most $\lambda$ clients. This problem has been well studied, a $2.5-\lambda^{-1}$-approximation is known for general graphs~\cite{haimovich1985bounds}, while Becker et al. have reported a PTAS on planar graphs recently~\cite{BeckerKS19}. Capacitated versions of other classical combinatorial optimization problems are very well studied and we give a highly non-exhaustive list here. Capacitated  unweighted vertex cover admits a $2$-approximation~\cite{Kao17,Wong17}. Capacitated versions of clustering problems like k-center~\cite{AnBCCGMS15}, k-median~\cite{ByrkaRU16} and scheduling problems~\cite{SahaS18} have also been widely popular.

\subsection{High Level Ideas and Techniques}

Most algorithms for either $\MMTC$ or rooted $\MMTC$ build upon the following idea. Assume that we know the optimal tree cost - say $\opt$. Further, for ease of exposition, let us assume that the graph is connected and does not contain any edge of length more than $\opt$. Now consider the $k$ trees in the optimal solution, each of total length at most $\opt$. Adding at most $k-1$ edges to the union of these trees forms a spanning tree of the entire graph. Hence, the minimum spanning tree (MST) cost is upper bounded by $(2k-1)\opt$. This leads to the natural idea of starting with the MST of the graph. At a high-level, we can root the tree at an arbitrary vertex, traverse it bottom-up and chop it off as soon as the total cost of the traversed sub-tree is between $[2\opt, 4\opt]$. This needs to be done carefully and we would refer the reader to~\cite{EvenGKRS04} for details. Now, this can create at most $k$ partitions of the MST, each of cost at most $4\opt$. The rooted version uses a similar idea, but requires more care. Let us try to apply this idea to the capacitated case. A potential problem is that, we have no control over the structure of the MST. In particular, some part of the MST might be dense - it might contain a connected subtree that has small length but covers a large number of vertices, while some other parts might be sparse. Hence, cutting off the MST on the basis of length as above might end up producing infeasible trees, although the cost might be bounded. One possible idea to fix this could be to further cut off the dense subtrees and try and re-combine them with the sparse subtrees. However, it is not clear how to avoid either combining more than a constant number of subtrees or subtrees that are more than $\bigO(\opt)$ distance away from each other and hence cannot lead to a constant factor approximation. 

\noindent
{\bf Our Approach : Using $\lambda$-Steiner Trees.}
In order to develop the intuition for our core ideas, we focus on the special case where each tree in the optimal solution {\em covers} exactly $\lambda$ vertices ( note that each tree may span more than $\lambda$ vertices). We take a different approach to the problem by utilizing the concept of $\lambda$-Steiner Trees. Given a graph $G$ and a subset of vertices $R$ called {\em terminals}, a $\lambda$-Steiner tree on $R$ is a minimum length subtree of $G$ that contains exactly $\lambda$ vertices from $R$. We begin with the observation that each tree in the optimal solution covers $\lambda$ vertices and hence $\OPT$ can be thought of as an union of $k$ $\lambda$-Steiner trees in $G$, although not necessarily of the minimum possible cost. A natural algorithm is to pick an arbitrary root vertex and construct a $\lambda$-Steiner tree. Computing the $\lambda$-MST and hence $\lambda$-Steiner trees is NP-Hard and hence we resort to the $4$-approximation algorithm that essentially follows from Garg's algorithm~\cite{Garg05}. If we are lucky, we might end up capturing an optimal tree and continue. However, in the unlucky case, the $\lambda$-Steiner tree might cover vertices from several of the optimal trees. Now if we disregard these vertices in further iterations and try to build another $\lambda$-Steiner Tree on the uncovered vertices, we might be stuck since such vertices in the union of the optimal trees might be far away from each other. Hence we cannot guarantee that a low cost $\lambda$-Steiner tree exists on these vertices. 

We fix this problem by being less aggressive in the first step. We try to build as many $\ceil*{\lambda/2}$-Steiner trees as possible that have cost at the most $\bigO(\opt)$ using Garg's algorithm - call such trees {\em good}. At some point, we might be left with vertices such that there does not exist any $\ceil*{\lambda/2}$-Steiner trees of low cost that can cover them - let us call them {\em bad vertices}. In order to cover the bad vertices, we deploy an iterative clustering procedure. We begin by applying the algorithm of Khani and Salavatipour~\cite{KhaniS14} for $\MMTC$, henceforth termed as the KS-algorithm, on the bad vertices. Note that this will return at most $k$ trees each of cost $\bigO(\opt)$, although we still cannot prove any lower bound on the number of vertices covered by each tree. Next we set up a bipartite matching instance with these newly formed trees on the left hand side and the good trees on the right hand side. We introduce an edge between two trees if and only if they are separated by a distance of at the most $\opt$. The crucial claim now is - if there exists a Hall Set in this matching instance, say $S$, then the number of trees in $S$ is strictly greater than the number of trees that the optimal solution forms with the vertices covered by $S$. Hence, applying the KS-algorithm on the vertices in $S$ reduces the number of trees in $S$ without increasing the cost of each tree. This idea forms the heart of our algorithm. We apply the KS-algorithm iteratively until there is no Hall Set, at which point we can compute a perfect matching of the bad trees. Each bad tree can now be combined with a good tree to produce a tree that has cost $\bigO(\opt)$ and contains at least $\lambda/2$ vertices each. In a nutshell, the above procedure circumvents the problem of creating too many sparse trees. It ensures that every tree is sufficiently dense - covers at least $\ceil*{\lambda/2}$ vertices and are of low cost. Consider a modified graph $\tilde{G}$ created by contracting the edges of the dense trees. Since the original graph is connected and each edge has length at the most $\opt$, $\tilde{G}$ is also connected and any edge in $\tilde{G}$ has length at the most $\opt$. We utilize these properties to distribute the vertices and ensure that the final set of trees have exactly $\lambda$ vertices each and cost at the most $\bigO(\opt)$. We note that under the assumption that every optimal solution tree has exactly $\lambda$ vertices as well, the above algorithm will produce exactly $k$ trees.

Our algorithm for the capacitated problem builds upon the above ideas. However, one major bottleneck is that we can no longer assume that each tree in the optimal solution contains exactly $\lambda$ vertices. In fact, there could be trees which have a small number of vertices, say less than $\lambda/2$ - call them {\em light} trees and the rest {\em heavy}. Handling this situation requires more subtle ideas. We again start by creating $\ceil*{\lambda/4}$-Steiner trees of low cost as long as possible which we call good trees and as before, we shall be left with some bad trees that contain less than $\lambda/4$ vertices each, but have bounded cost. Unfortunately, the iterative refinement procedure is no longer guaranteed to produce a perfect matching, as before. However, we are able to show the following. Existence of a Hall Set even after applying the said refinement is a certificate of the fact that optimal solution contains a significant number of light trees - suppose this number is $k_{\ell}$. Then, we add $k_{\ell}\cdot \ceil*{\lambda/2}$ dummy vertices co-located with suitably chosen bad vertices. The upshot is that, this addition ensures that each bad tree now becomes a good tree, together with the dummy vertices. Further, the total number of vertices including the dummy vertices is still bounded above by $k\lambda$. Together, this gives us that, creating trees of size exactly $\lambda$ cannot result in more than $k$ trees in the solution.

%\begin{lemma}
%	\label{lem:2mst}
%	Let $G=(V,E)$ be an undirected weighted connected graph with edge cost $\ell:E\leftarrow \Real_{\geq 0}$. Then, for any subset $V'\subseteq V$, the cost of $\mst(V')$ is at the most $ \leq 2\ell(T(V'))$, where $T(V')$ is any sub-tree in $G$ that spans $V'$.
%\end{lemma}
%
%\begin{proof}[sketch]
%	Consider the tree $T(V')$. Make a multigraph $\hatT$ by taking two copies of each edge in $T(V')$. $\hatT$ is Eulerian since all degrees are even. Now use the Eulerian walk, to visit all the vertices, short-cutting whenever a vertex is repeated. The cost of this walk is upper bounded by $2T(V')$ and hence the MST cost on the induced metric  of $V'$ is upper bounded by the same quantity.
%\end{proof}

%We shall make a final simplifying assumption. Different trees in a solution of the $\capMMTC$ or $\caprMMTC$ problem are allowed to share edges and even vertices. However, we assume that the solution partitions the set of vertices in to $k$ disjoint subsets and the final trees in any solution is the MST on the induced subgraph of a particular partition. Applying Lemma~\ref{lem:2mst}, we lose, at the most, a factor $2$ due to this assumption. We make this simplification purely due to ease of exposition of our main ideas and can be easily removed. 

\section{Capacitated Min-Max Tree Cover}
\label{sec:capacitated}
In this section, we shall describe our algorithm for the $\capMMTC$ problem and prove Theorem~~\ref{thm:capmmtc}. The first step in our algorithm is to guess the value of the optimal solution - call it $\opt$. We remove all edges from $G$ that are of length bigger than $\opt$ since the optimal solution can never use any such edge. The resulting graph $\hat{G}$ has, say, $p$ connected components - call them $G_1, G_2, \cdots G_{p}, p \leq k$. Suppose, in the optimal solution, there are $k_i$ trees that cover all vertices in $G_i, i=1,2\cdots p$. Then, $\Sigma_{i=1}^{p} k_i =k$. For each connected component $i$, we run our algorithm to get at most $k_{i}$ trees with cost $\bigO(T^{*})$. Due to the above argument, in subsequent exposition, we shall assume that we have a connected graph $G(V,E)$ with edge lengths $\ell(e) \leq \opt, \forall e\in E$ and there are $k$ trees in the optimal solution that cover $V$ such that every tree has cost at most $\opt$ and covers at most $\lambda$ vertices. There could be multiple optimal solutions. For the purpose of our discussion we pick one arbitrarily and whenever we refer to an optimal solution, we mean this particular solution. We shall divide the trees in the optimal solution into two classes for the purpose of analysis. Define a tree $T$ in the optimal solution to be {\em light} if $|\cov(T)| \leq \floor*{\lambda/2}$ and {\em heavy} otherwise. Define $k_{heavy}$ and $k_{light}$ to be the number of heavy and light trees respectively. We shall prove the following theorem.
%Further, let us re-define the distance metric by computing the metric closure of edges in $G'$ as follows.
%
%\begin{align*}
%\ell'(e) &= \ell(e) , e\in E'\cap E \\
%         &= \text{shortest path distance between $u,v$ in }G', e=(u,v)\in E\setminus E'
%\end{align*}
\begin{theorem}
	\label{thm:capMMTCg}
	
	Given a connected graph $G(V,E)$ with the edge lengths $\ell:E\rightarrow\Real_{\geq 0}$, $\ell(e)\leq \opt, \forall e\in E$ and non-negative integers $k,\lambda$, suppose there exists $k$ trees $T_1,T_2,\dots,T_k$ along with an assignment of vertices to the trees, such that $T_{i}$ is a subgraph of $G$, $|\cov(T_{i})| \leq \lambda$ and $\ell(T_{i}) \leq \opt$ for $i \in [k]$. Then there exists a polynomial time algorithm that finds a set of trees $T'_1, T'_2, \cdots T'_{k'}$ along with an assignment of vertices to the trees such that $k'\leq k$, for each $T'_j$, $\ell(T'_j) = \bigO(\opt)$ and $|\cov(T'_j)| \leq \lambda$.   
	
\end{theorem} 

Before proceeding to the proof of the Theorem above, we show how it implies a proof for Theorem~\ref{thm:capmmtc}.

\begin{proof} [{\bf Theorem~\ref{thm:capmmtc}}]
	We can use Theorem~\ref{thm:capMMTCg} to carry out a binary search for the correct value of $\opt$ - the optimal solution - in the range $[0,\Sigma_{e\in E}\ell(e)]$. For a particular choice of $\opt$, we remove all edges that are of length more than $\opt$ and apply Theorem~\ref{thm:capMMTCg} to each of the connected components. If the total number of trees formed by our algorithm over all components is at most $k$, then we iterate with a guess $\opt/2$ , otherwise with $2\opt$. The correctness follows from Theorem~\ref{thm:capMMTCg}, since our algorithm will create at most $k'$ trees for connected component $V'$, provided optimal solution makes $k'$ trees as well and the guessed value $\opt$ is correct.
\end{proof}  

\begin{definition}
	Given an integer $\lambda > 0$, a tree $T$ is called $\lambda$-good if $|\cov(T)| \geq \lambda/4$.
\end{definition}

\begin{definition}
	{\bf $\lambda$-Steiner Trees.} Given a graph $G(V,E)$ and a subset of vertices $R$ called {terminals},  $\lambda$-Steiner tree is a tree which is a subgraph of $G$ that spans exactly $k$ terminals. The special case of $R=V$ is called $\lambda$-MST.  
\end{definition}
We shall be using the $2$-approximation algorithm for minimum cost $\lambda$-MST from~\cite{Garg05} as a black-box, which implies a $4$-approximation for the minimum cost $\lambda$-Steiner tree problem. Our algorithm has three phases $\calA$, $\calB$ and $\calC$, each making progress towards the proof of Theorem~\ref{thm:capMMTCg}.

\subsection{Algorithm $\calA$ : Constructing $\lambda$-good trees }

Recall that the input is a  graph $G=(V,E)$, with edge lengths $l_e\leq \opt, e\in E$. Further, we assume that there exists a partition of $V$ into $k$ trees such that the cost of each tree is at most $\opt$ and each tree covers at most $\lambda$ vertices (with respect to some assignment). We maintain a set of covered vertices $V_c$ throughout our algorithm. Initially $V_c=\phi$. The following is done iteratively. Pick an arbitrary vertex $v \in V \setminus V_c$ which has not been considered in any previous iteration. Construct a $4$-approximate minimum cost $\ceil*{\lambda/4}$-Steiner Tree in $G$, rooted at $v$ with the terminal set being  $V\setminus V_c$. If the cost of this tree is at most $4\opt$, then add this tree to the set $\lambda_{good}$, add all terminals spanned by this tree to $V_c$. At the termination of this loop, set $V_{bad}=V\setminus V_c$.

\begin{figure}[ht]
	\begin{center}
		\begin{boxedminipage}{5.5in}
			{\bf Input}: Graph $G=(V,E),\lambda$ \\
			{\bf Output}: Set of trees $\lambda_{good} $ and set of vertices $V_{bad}$ \\ 
			
			\sp \sp $V_c \leftarrow \nullset, \lambda_{good}\leftarrow \nullset, V_{temp}\leftarrow V$ \\
			\sp \sp While $V_{temp} \neq \nullset$ \\
			\sp \sp \sp \sp Pick $v\in V_{temp}$ arbitarily \\
			\sp \sp \sp \sp $T \leftarrow \ceil*{\lambda/4}$-Steiner tree rooted at $v$ with terminal set defined as $V\setminus V_c$\\
			\sp \sp \sp \sp $\cov(T) \leftarrow$ terminals in $V\setminus V_c$ covered by $T$ \\
			\sp \sp \sp \sp {\bf If} $\ell(T) \leq 4\opt$ {\bf Then} \\
			\sp \sp \sp \sp \sp \sp Add $T$ to $\lambda_{good} $ \\
			\sp \sp \sp \sp \sp \sp Update $V_{temp}\leftarrow V_{temp}\setminus \cov(T), V_c\leftarrow V_c \cup \cov(T)$ \\
			\sp \sp \sp \sp {\bf Else} \\
			\sp \sp \sp \sp \sp \sp Update $V_{temp}\leftarrow V_{temp}\setminus \{v\}$  \\
			\sp \sp $V_{bad} \leftarrow V\setminus V_c$ \\
			\sp\sp {\bf Return} $\lambda_{good}$ and $V_{bad}$
		\end{boxedminipage}
		\caption{Algorithm $\calA$ }
		\label{alg:AlgA}
	\end{center}
\end{figure}

\begin{lemma}
	\label{lem:AlgA}
	At termination of Algorithm $\calA$, 
	\begin{enumerate}
		\item $V = \cov(\lambda_{good})\cup V_{bad}$, where $\cov(\lambda_{good})=\bigcup_{T \in \lambda_{good}}\cov(T) $
		\item For any tree $T \in \lambda_{good}$, $|\cov(T)| = \ceil{\lambda/4}$ and $\ell(T) \leq 4\opt$
	\end{enumerate} 
\end{lemma}

\begin{lemma}
	\label{lem:atLeastHalf}
	For any heavy tree in $\OPT$, say $T$, at least half of the vertices covered by $T$ in $\OPT$, ie. $\cov(T)$ are covered by the union of trees in $\lambda_{good}$
\end{lemma}

\begin{proof}
	Assume otherwise. This implies that at least $\lambda/4$ vertices covered by $T$ are added to $V_{bad}$ - let us call this set $V'$. Now, $\ell(T) \leq \opt$ and there must exist a connected sub-tree of $T$ that covers exactly $\ceil{\lambda/4}$ vertices from $V'$. Hence,  some iteration of Algorithm $\calA$ must have returned a tree of cost at most $4\opt$ covering $\ceil{\lambda/4}$ vertices from $V'$ and this tree was added to $\lambda_{good}$. Consequently, $V'$ cannot be a subset of $V_{bad}$ leading to a contradiction.  
\end{proof}

\subsection{Algorithm $\calB$ : Covering $V$-bad vertices with $\lambda$-good trees}

Recall that in the previous section, we gave algorithm to construct a set of trees $\lambda_{good}$ such that for any tree $T\in \lambda_{good}$, $\ell(T) \leq 4\opt$ and $|\cov(T)| = \ceil*{\lambda/4}$. Further, we extracted a subset of vertices $V_{bad}$ such that no $\ceil{\lambda/4}$-Steiner tree of cost at the most $4\opt$ exists covering these vertices. In this section, we shall give a procedure to cover the vertices in $V_{bad}$, either by constructing new trees that have large size but low cost, or by merging them in to the existing $\lambda$-good trees. In the following procedure, we shall be making use of the 3-approximation algorithm by Khani and Salavatipour~\cite{KhaniS14} for the $\MMTC$ problem. Recall the main theorem from above paper, re-phrased to suit our notations.

\begin{theorem}~\cite{KhaniS14}
	\label{thm:KS}
	Given an undirected graph $G(V,E)$, an integer $k$ and edge lengths $\ell$, if there exists a set of $k$ trees (subgraphs of $G$) $T_1, T_2, \cdots T_{k}$, along with an assignment of vertices to the trees, such that $\max_{i\in [k]} \ell(T_i) \leq \opt$ and $\cup_{i=1}^{k}\cov(T_i)=V$, then there exists a polynomial time algorithm that outputs trees (subgraphs of G) $T'_1, T'_2,\cdots T'_{k'}$, along with an assignment of vertices to the trees such that $\max_{j\in [k']} \ell(T'_{j}) \leq 3\opt$,  $\cup_{j=1}^{k'}\cov(T'_j)=V$ and $k'\leq k$.
\end{theorem}

We refer to the algorithm referred to in the above theorem as the KS-Algorithm in subsequent sections. As mentioned earlier, although the trees $T'_j$ might share vertices and edges, the set of vertices that such trees {\em cover} forms a partition of $V$. Note that KS algorithm takes an edge weighted graph and number of trees as input. 

\noindent
{\bf A Bipartite Matching Instance:} Our algorithm first constructs a maximum matching instance on a suitably defined bipartite graph $H = (B,Q,F)$. In the course of algorithm, $Q$ remains fixed. However, the other partitite set $B$ and consequently the edge set $F$ is refined iteratively, each time giving rise to a new maximum matching instance. Our invariants will ensure that, either we end up with a perfect matching of $B$ or conclude that there are significantly many light trees in the optimal solution, which need to be handled separately.  

For clarity of notation, we parameterize set $B$ by iteration index $t$ - define $B_t$ to be the partite set $B$ at iteration $t$. A vertex in the set $B_t$ represents a subtree $T_{j,t}, j=1,2,\cdots |B_t|$ such that $\cup_{j=1}^{|B_t|} \cov(T_{j,t})= V_{bad}$. The set $Q$ contains a vertex corresponding to each tree in $T\in \lambda_{good}$. Next, we define the edges in $F_t$ at any iteration $t$. There exists an edge $e=(T_{j,t}, T)\in F_t,T \in Q$ if and only if $d(V(T_{j,t}),V(T)) \leq \opt$. For a fixed $t$, define $\calM_t$ to be a maximum matching instance on the graph $(B_t, Q, F_t)$. For any subset $B'\subseteq B_t$, let $\calN(B')$ be the neighborhood set of $B'$ in $ Q$.  

\noindent
{\bf Algorithm IterRefine.} We initialize by setting $B_0 = V_{bad}$ - each vertex in $B_0$ is a trivial tree containing a single vertex. Now we have the required setup to describe the iterative refinement procedure. At any iteration $t$, we solve the maximum matching instance $\calM_t$. If $\calM_t$ admits a perfect matching of $B_t$, then we terminate. Otherwise we consider the Hall Set $S\subseteq B_t$ with maximum deficiency, that is $\argmax_{S\subseteq B_t}(|S| - |\calN(S)|)$. Let set $S$ contains the subtrees $T_{1,t}, T_{2,t}, \cdots T_{s,t}, s=|S|$. Let $G_{t}=(V_{t},E_{t})$ be defined as: $V_{t}= \cup_{j=1}^{s} \cov(T_{j,t}), E_{t}= \{(u,v)| u,v \in V_{t}\}$ and $\ell(u,v)$ is the length of shortest path between $u$ and $v$ in $G$. We run the KS algorithm on $G_{t}$ with $k=s-1$. Suppose it returns trees $T'_{1,t}, T'_{2,t}, \cdots T'_{s',t}$, $s' \leq s-1$. If cost of each such tree is at most $6T^{*}$, we define $B_{t+1} = (B_t \setminus S) \bigcup_{j=1}^{s'}\{T'_{j,t}\}$, the edge set $F_{t+1}$ accordingly and proceed to the next iteration, else we terminate. 

\begin{figure}[ht]
	\begin{center}
		\begin{boxedminipage}{5.5in}
			{\bf Input }: $\lambda_{good}, V_{bad}$ \\
			{\bf Output}: A maximum matching in $H_t$ and Unmatched Set $\overline{B_{m}}$ \\ 
			
			\sp \sp Initialize the bipartite graph $H_0=(B_0,Q,F_0), t\leftarrow 0$ \\
			\sp \sp {\bf While} there exists a Hall Set in $B_t$ \\
			\sp \sp \sp \sp Pick the Hall Set $S=\argmax_{S\subseteq B_t}(|S|-|\calN(S)|)$ \\
			\sp \sp \sp \sp Apply KS-algorithm on $G_{t}$ with $k=s-1$. Let $T'_{1,t}\cdots T'_{s',t}$ be the output of the algorithm with $s' \leq s-1$.\\
		
			\sp \sp \sp \sp {\bf If} $\max_{i=1}^{s'}\ell(T'_{i,t}) >  6T^{*}$ {\bf Then} STOP \\
			\sp \sp \sp \sp {\bf Else} 
			\{ $B_{t+1}\leftarrow (B_t\setminus S) \bigcup_{j=1}^{s'}\{ T'_{j,t} \}$ \sp  and \sp $t\leftarrow t+1$ \} \\
			\sp\sp {\bf Return} Maximum Matching $\calM_t$ and unmatched subset $\overline{B_{m}}\subseteq B_t$
		\end{boxedminipage}
		\caption{Algorithm IterRefine }
		\label{alg:AlgA}
	\end{center}
\end{figure}

The following lemmas are crucial to prove the correctness of our algorithm. We define the projection of the optimal solution on $S$ as follows.
$$ \restr{\OPT}_{S} = \{T'\in \OPT : \exists T_{i,t}\in S, \cov(T')\cap \cov(T_{i,t}) \neq \nullset\} $$

\begin{lemma}
	\label{lem:ks}
	At any iteration $t$, let the number of trees in $\restr{\OPT}_{S}$ be $m$. Then there exist $m$ trees $T_{i}, i \in [m]$ such that $\cup_{i=1}^{m}V(T_{i})=V_{t}, \ell(T_{i}) \leq 2T^{*}$ and $T_{i}$ is a subtree of $G_{t}$ for $i \in [m]$.
\end{lemma}

\begin{proof}
    Let $T'_{i},i \in [m]$ be the trees in $\restr{\OPT}_{S}$. Fix a tree, say $T'_{i}$. Consider the graph with each edge of $T'_{i}$ doubled. It is connected and all degrees are even, hence there exists an eulerian walk, say $v_{1},v_{2},\ldots,v_{l}$. Note that vertices in this walk may be repeated. Identify a unique vertex for each $v \in V(T'_{i}) \cap V_{t}$ in this sequence. This defines a subsequence of the walk, say $v_{i_{1}},v_{i_{2}},\ldots,v_{i_{l}}$ with every vertex in $V_{t} \cap V(T'_{i})$ occurring exactly once. This sequence defines a path in $G_{t}$ since $G_{t}$ is a complete graph. The total length of path is at most $2T^{*}$, as length of any edge in $G_{t}$ is the shortest path distance between the end points in $G$ and total length of the walk is at most $2T^{*}$. Repeating this argument for all trees in $\restr{\OPT}_{S}$, we get that vertices in $V_{t}$ can be covered by $m$ trees of length at most $2T^{*}$ each.
\end{proof}

%\begin{proof}
%	We prove the lemma by induction on the number of iteration. Initially, $B_0$ contains singleton trees and hence the lemma holds trivially.
	
%	Let the lemma be true for all iterations till $t$. Now, consider iteration $t+1$. The trees in $B_{t+1}\setminus B_t$ are obtained by running the KS-algorithm on the set of vertices covered by the trees belonging to some Hall Set $S$ in the maximum matching instance $\calM_t$. Consider the trees in $\restr{\OPT}_{S}$ that covers these vertices in $\OPT$. Any such tree has cost at the most $\opt$. Applying Theorem~\ref{thm:KS}, the cost of any tree output by the KS-algorithm is at the most $3\opt$. 
%\end{proof}

\begin{lemma}
	\label{lem:hall-set}
	Suppose $S$ is the maximum deficiency Hall Set, if any, at termination of Algorithm IterRefine. Further, let $k_{\ell}$ be the number of light trees in $\restr{\OPT}_{S}$. Then $|S| - |N(S)| \leq k_{\ell}$.
	
\end{lemma}

\begin{proof}
	Let $k_h$ be the number of heavy trees in $\restr{\OPT}_{S}$. We first claim that $|S| \leq k_h + k_{\ell}$. Assume otherwise. But then,  by Lemma \ref{lem:ks}, there exists a partition of $S$ into exactly $k_h + k_{\ell}$ trees in $G_{t}$ with cost at most $2T^{*}$, by Theorem~\ref{thm:KS}, the KS-algorithm applied on the Hall Set $S$ with $k=|S|-1$ must have produced trees with cost at most $6T^{*}$ and the algorithm would not have terminated. 
	
	Now, we prove that $|N(S)| \geq k_h$. Consider any heavy tree in $\restr{\OPT}_{S}$, say $\hatT$. By Lemma~\ref{lem:atLeastHalf}, at least half of the vertices in $\hatT$ are covered by $\lambda$-good trees. By construction of the bipartite graph $(B_t, Q)$, all such $\lambda$-good trees belong to the set $N(S)$, since any two vertices in the same tree of the optimal solution can be at the most $\opt$ distance away from each other. Thus, adding up over all heavy trees in $\restr{\OPT}_{S}$, the total number of vertices covered by the $\lambda$-good trees in $N(S)$ is at least $k_h\cdot \ceil{ \lambda/4}$. But, each $\lambda$-good tree in $N(S)$ covers exactly $\ceil{\lambda/4}$ vertices. Hence, by pigeon-hole principle, $|N(S)| \geq k_h$. 
	Combining the above two facts, we have $|S| - |N(S)| \leq k_{\ell}$.
\end{proof}

\begin{corollary}
	\label{cor:hall-set}
	Let $t'$ be the final iteration of the Algorithm IterRefine. Then, the total number of unmatched trees in $B_{t'}$ is at most $k_{light}$.
\end{corollary}

\begin{proof}
	We use the fact that the size of the maximum matching in $B_{t'}$ is exactly equal to $|B_{t'}| - (|S| - |N(S)|)$, where $S$ is the maximum deficiency Hall Set~\cite{Diestel12}. Hence, the total number of unmatched trees in $B_{t'}$ is at exactly $|S|-|N(S)| \leq k_{\ell} \leq k_{light}$, by Lemma~\ref{lem:hall-set}.
\end{proof}

We proceed to the final phase of our algorithm, equipped with the above Lemmas. Recall that, there are two types of partitions in $B_{t'}$ - either matched or unmatched by $\calM_{t'}$. First, let us consider the matched partitions - call this set $B_{m}$. For each tree $T$ in $B_m$, we take its union with the $\lambda$-good tree in $Q$, say $\hatT$ that it is matched to in $\calM_{t'}$. More formally, we build a tree $T'$ by connecting $T,\hatT$ using the shortest path between the closest pair of vertices in $T$ and $\hatT$. We set $\cov(T')=\cov(T) \cup \cov(\hatT)$. We remove $\hatT$ and add the newly constructed tree $T'$ to the set $\lambda$-good. 

Finally, we consider the set of unmatched trees $\overline{B_{m}}$. For each such tree, we add $\ceil*{\lambda/2}$ dummy vertices co-located with any arbitrary vertex in the tree. We add this tree to $\lambda_{good}$. This completes the description of the algorithm.

\begin{lemma}
	\label{lem:goodTrees}
	At termination, Algorithm $\calB$ returns a set of trees $\lambda_{good}$ that 
	\begin{enumerate}
		\item covers all vertices in the  set $V \cup V_{dummy}$.
		\item The cost of any tree in $\lambda_{good}$ is at most $11\opt$. 
		\item $|V \cup V_{dummy}| \leq k\lambda$
	\end{enumerate} 
\end{lemma}

\begin{proof}
	First we show that all trees in $\lambda_{good}$ are indeed $\lambda$-good, that is, they cover at least $\lambda/4$ vertices. Recall that, after algorithm $\calA$, the set $Q$ contains only $\lambda$-good trees, by Lemma~\ref{lem:AlgA}. At termination of algorithm $\calB$, each tree is either formed by joining of a tree in the set $Q$ with some tree in $B_m$. Hence, the tree continues to be $\lambda$-good. The only other type of tree in $\lambda_{good}$ are formed by adding $\ceil*{\lambda/2}$ dummy vertices to an unmatched tree in $\overline{B_{m}}$. Hence, they are $\lambda$-good by construction. Property 1 follows from Lemma~\ref{lem:AlgA} and the fact that cover of trees in $B_{t'}$ is a partition of the vertices in $V_{bad}$. Next we prove Property 2. There are two types of trees in $\lambda_{good}$. Let us consider any tree that is formed by the union of some matched tree in $B_{m}$, say $\calT$ and some already existing $\lambda$-good tree in $Q$, say $\hatT$. By Lemma~\ref{lem:AlgA}, $\ell(\hatT) \leq 4\opt$. By Lemma~\ref{lem:ks}, $\ell(\calT)\leq 6\opt$. A tree is formed by joining these two trees by a path of cost at the most $\opt$. So, the resultant tree has cost at the most $(4+6+1)\opt = 11\opt$. The second kind of tree corresponds to the unmatched vertices in $B_{t'}$ and hence their cost is bounded above by $3\opt$, by Lemma~\ref{lem:ks}. Addition of dummy vertices does not affect the cost at all.
	
	Property 3 can be proved as follows. Recall that $|V| \leq k_{heavy}\lambda + k_{light}\lambda/2 $. Now, we add $\lambda/2$ dummy vertices to each unmatched partition in $B_{t'}$, which, by Corollary~\ref{cor:hall-set}, can be at the most $k_{light}\lambda/2$. Hence,
	$$|V\cup V_{dummy}| \leq k_{heavy}\cdot\lambda + 2\cdot k_{light}\cdot \lambda/2 \leq k\lambda$$ \end{proof}

The above lemma guarantees that the total number of vertices in the graph is at the most $k\lambda$. For technical reasons that would be clear in Section~\ref{sec:exact}, we add some more dummy vertices co-located with an arbitarily selected $v\in V$, such that the total number of vertices is an integer multiple of $\lambda$. We note at this point that we do not make any claim about the number of trees in the set $\lambda_{good}$. In fact, the number can be larger that $k$, the original budget. However, in the next Section, we show how to refine the trees such that number of trees is at most $k$, without increasing the cost by more than a constant factor and maintaining the capacity constraints. Also note that all our procedures are polynomial time.

%% file: exact.tex
\subsection{Algorithm $\calC$ : Converting $\lambda$-good trees to feasible trees}
\label{sec:exact}

In the previous section, we gave algorithms that ensure the property that each tree covers at least $\lambda/4$ vertices and the cost of each tree is bounded by $\bigO(\opt)$. In this section, we prove that it is possible to construct a solution that creates at most $k$ admissible trees given the previous solution, where $k$ is the number of trees that the optimal solution uses to cover the vertices in the connected component $G=(V,E)$. We prove the following theorem in this section.

\begin{theorem}
	\label{thm:exact}
	Let $T_{1},T_{2},\ldots,T_{m}$ be trees returned by the Algorithms $\calA$ and $\calB$ such that $|\cov(T_{i})| \geq \lambda / c$, where $c\geq 1$, $\ell(T_{i}) \leq \alpha T^{*}$ for $1 \leq i \leq m$. If $\lambda$ divides $|\cup_{i=1}^{m}\cov(T_{i})|$, then there exist trees $T'_{1},T'_{2},\ldots,T'_{l}$ spanning $V$ such that $|\cov(T'_{i})|=\lambda, \ell(T'_{i}) \leq (6c+6c \alpha-2\alpha -6) T^{*}$ for $1 \leq i \leq l$. Further, such trees can be found in polynomial time.
\end{theorem}

Define $G_{P}$ be to be an unweighted graph with $T_{1},T_{2},\ldots,T_{m}$ as vertices. There is an edge $(T_{i},T_{j})$ if $d(T_{i},T_{j})=d(V(T_{i}),V(T_{j})) \leq T^{*}$. 

\begin{claim}
	$G_{P}$ is connected. Further, the implicit cost of each unweighted edge in $G_p$ is at most $\opt$.
\end{claim}

\begin{proof}
	The claim follows from the fact that we run Algorithms $\calA$ and $\calB$ only on connected components of $\hat{G}$ formed after removing all edges that are of length more than $\opt$ from $G$.
\end{proof}

We will need the following theorem by Karaganis \cite{karaganis1968cube} to prove our result:
\begin{theorem} \cite{karaganis1968cube}
	Let $G=(V,E)$ be a connected simple unweighted graph. Define $G^{3}$ to be the graph on the same set of vertices and set of edges 
	$$E^3 = \{(u,v)| \text{shortest path distance between u and v in G is at the most 3}\}$$
	Then, $G^3$ has a Hamiltonian path.
\end{theorem}

By the above theorem, there exists a Hamiltonian path in $G_{P}^{3}$. After renaming, let $T_{1},T_{2},\ldots,T_{m}$ be the order in which the vertices of $G$ appear on that path. Then $d(T_{i},T_{i+1}) \leq 3T^{*}+2 $(max diameter of any $T_{i}$) $\leq (2\alpha+3)T^{*}$ for $1 \leq i \leq m-1$. We are now ready to prove the main results of this section.
\begin{lemma}\label{big}
	Let $T_{1},T_{2},\ldots,T_{m}$ be trees such that $ |\cov(T_{i})| \geq \lambda /c $ for some $c > 1$, $\ell(T_{i}) \leq \alpha T^{*}$ for $1 \leq i \leq m$ and  $d(T_{i},T_{i+1}) \leq D$ for $1 \leq i \leq m-1$. There exist trees $T'_{1},T'_{2},\ldots,T'_{l}$ such that $d(T'_{i},T'_{i+1}) \leq D$, $|\cov(T'_{i})| \geq \lambda$ for $1 \leq i \leq l-1$ and $\ell((T'_{i}) \leq (c-1)D+ c \alpha T^{*}$ for $1 \leq i \leq l$. Further, such trees can be computed in polynomial time.
\end{lemma}
\begin{proof}
	We make $l=\lceil m/c \rceil$ trees. Let $P(T_{p},T_{q})$ denote the shortest path connecting trees $T_{p}$ and $T_{q}$. $T'_{i}=\bigcup_{j=(i-1)c+1}^{min(ic,m)}(T_{j} \cup P(T_{j},T_{j+1}))$ for $1 \leq i \leq l$. Since $T'_{1},T'_{2},\ldots,T'_{l-1}$ are formed by combining $c$ trees, each covering at least $\lambda /c$ vertices, $|\cov(T'_{i})| \geq $ $\lambda$ vertices. Since any $V'$ is formed by combining at most $c$ trees, $\ell(T') \leq (c-1)D + c \alpha T^{*}$. Also, by construction $d(T'_{i},T'_{i+1}) \leq D$ for $1 \leq i \leq l-1$
\end{proof}
\begin{lemma}\label{exact}
	Let $T_{1},T_{2},\ldots,T_{m}$ be trees such that $ |\cov(T_{i})| \geq \lambda$ for $2 \leq i \leq m$, $\ell(T_{i}) \leq \alpha T^{*}$ for $1 \leq i \leq m$ and  $d(T_{i},T_{i+1}) \leq D$ for $1 \leq i \leq m-1$. If $\lambda$ divides $|\cup_{i=1}^{m}\cov(T_{i})|$, then there exist trees $T'_{1},T'_{2},\ldots,T'_{l}$ such that $|\cov(T'_{i})|=\lambda, \ell(T'_{i}) \leq D + 2\alpha  T^{*}$ for $1 \leq i \leq l$. Further, such trees can be found in polynomial time.
\end{lemma}
\begin{proof}
	We will prove this by induction on $m$. If $m=1$, then $\lambda$ divides $|\cov(T_{1})|$. We partition $\cov(T_{1})$ into sets of size $\lambda$ arbitrarily. Cost of a tree covering each of these partitions will be at most $\alpha T^{*}$  ($T_{1}$ is one such tree). Let $m>1$ and $|\cov(T_{1})|=c \lambda+p$, where $c\geq 0,1 \leq p \leq \lambda -1$. We first make $c$ partitions of size $\lambda$ each out of $\cov(T_{1})$. Cost of spanning tree covering each of these partitions is at most $\alpha T^{*}$. We create a new tree, say $T_{p}$, by taking $p$ remaining points of $\cov(T_{1})$ and any $\lambda-p$ points of $\cov(T_{2})$. $|\cov(T_{p})|=\lambda$ and $\ell(T_{p}) \leq  \ell(T_{1})+D+\ell(T_{2}) \leq  D+ 2 \alpha T^{*}$. The remaining points form a smaller instance and by induction can be covered by trees with desired properties. This completes the induction step and proof of the theorem.
\end{proof}

\begin{proof}[{\bf Theorem~\ref{thm:exact}}]
	From the discussion above and applying Lemma \ref{big}, we get trees covering  at least $\lambda$ vertices and cost at most $C=(c-1)D+c\alpha T^{*}$. Using Lemma \ref{exact} on these trees, we get trees covering exactly $\lambda$ vertices and cost at most $D+2C$. Plugging in $D=(2\alpha+3)T^{*}$ gives the desired result.
\end{proof}

\begin{proof} \textbf{[Theorem  \ref{thm:capMMTCg}]}
	Plugging in $c=4$ and $\alpha=11$ from Lemma \ref{lem:goodTrees} gives a $260$-factor approximation algorithm for the capacitated min max tree cover problem.
\end{proof}

%% file: rooted.tex
\section{Extension to Capacitated Rooted Tree Cover}
\label{sec:rooted}
In this section, we show how to extend Theorem~\ref{thm:capMMTCg} to the rooted setting. Recall that in $\caprMMTC$, in addition to the settings of $\capMMTC$, we given a set of root vertices, $R=\{r_1, r_2, \cdots r_k\}$. The $k$ output trees need to be rooted at each of these root vertices. We prove the following theorem.

\begin{theorem}
	\label{thm:caprmmtcg} 	
	Given a graph $G=(V,E)$ with the edge lengths $\ell:E\rightarrow\Real_{\geq 0}$, a set of roots $R \subseteq V, |R|=k$ and a parameter $\lambda$, suppose there exist subtrees of $G$, $T_1, T_2, \cdots T_{k'}$, where $\cup_{i=1}^{k} \cov(T_{i})=V$,$\max_{i\in[k']} |\cov(T_i)| \leq \lambda,  V(T_{i}) \cap R \neq \phi$ and $\ell(T_i) \leq \opt, \forall i=1,2,\cdots k$. Then there exists a polynomial time algorithm that finds subtrees of $G$, $T'_1, T'_2, \cdots T'_{k'}, k'\leq k$ such that $\cup_{i=1}^{k'} \cov(T'_{i})=V$, $\max_{i\in[k']} |\cov(T'_i)| \leq \lambda$, $V(T'_i) \cap R \neq \phi$  and $\ell(T'_i) \leq \beta\opt$, where $\beta = \bigO(1)$. 
\end{theorem}
\begin{proof} Similar to Section~\ref{sec:capacitated}, we focus on one of the connected components, say $G'$ obtained after deleting all edges of length larger than $\opt$. Let $R'=\{r_1, r_2, \cdots r_{k'}\}$ be the subset of roots belonging to this component. There exists a partition of $V'-R'$ into $k'$ trees, each covering at most $\lambda-1$ vertices and cost no more than $\opt$. Using Theorem \ref{thm:capMMTCg}, we first compute subtrees of $G'$ , $T_{1},\ldots, T_{k'}$, such that $|\cov(T'_{i})| \leq \lambda -1$ and cost is  $\bigO(\opt)$. We construct a bipartite graph $H'=(\calT \cup R', E')$, where $\calT$ contains a vertex corresponding to each tree $T'_j, j=1,2,\cdots k'$. We shall abuse notations slightly and interchangeably use $T'_j$ to denote the vertex in $\calT$ corresponding to the tree $T'_j$ and the actual tree itself, noting the difference wherever necessary. $(r_i,T'_j) \in E'$ if and only if $d(r_i,T'_j) \leq \opt$. If there exists a matching in $H'$ such that all vertices of $\calT$ are matched, we add the matching root vertex to $T'_j$. This increases the cost of a component by at most $\opt$. Now each component contains at least one root vertex, covers at most $\lambda$ vertices and costs $\bigO(\opt)$. If all the vertices of $\calT$ cannot be matched, there exists a set $S \subseteq T$ such that $|N(S)| < |S|$. Let $V(S)= \cup_{T'_{j} \in S}V(T'_{j})$. In optimal solution, every vertex of $V(S)$ is contained in a tree with one of the vertices in $N(S)$ as its root. Hence, there exists $|N(S)|$ trees such that each tree covers at most $\lambda-1$ points of $V(S)$, costs no more than $4\opt$ and union of all the trees cover $V(S)$. We use the result of Theorem \ref{thm:capMMTCg} to get at most $|N(S)|$ trees covering $V(S)$, such that each tree covers at most $\lambda-1$ vertices and costs $\bigO(\opt)$. Note that this procedure is similar to $\bf{IterRefine}$ defined earlier. We now have a new set of trees. We modify our bipartite graph $H'$ by removing vertices corresponding to trees in $S$ and replacing them by vertices corresponding to newly formed trees. We modify $E'$ suitably. If there exists a matching in the new graph, we are done, otherwise we repeat the above procedure till we find a perfect matching. Note that the size (in terms of number of vertices) of $H'$ decreases by at least $|S|-|N(S)| \geq 1$ after every step and hence will halt in at most $k'$ steps.  \end{proof}

%A very similar argument as in Section~\ref{sec:capacitated} shows that the above theorem implies Theorem~\ref{thm:caprmmtc} and we skip the formal proof here.
%

\section{Conclusion}
In this paper, we give the first constant factor approximation algorithms for the hard uniform capacitated versions of unrooted and rooted Min Max Tree Cover problem on a weighted graph. Our main technical contribution lies is devising an iterative refinement procedure that distributes vertices as evenly as possible among the different trees in the cover, without incurring too much cost. 

Our current approximation factor is large ($>250$) and although we did not attempt to optimize, it is unlikely that our approach can reduce the factor to a very small constant. Such a result requires new ideas. Perhaps somewhat surprisingly, no LPs with constant integrality gap is known, even for the uncapacitated problem. We leave open the question of finding such an LP, in particular, configuration-style LPs, which can possibly lead to a small approximation ratio. On the other hand, the only known hardness result, to the best of our knowledge, follows from the hardness of the uncapacitated version~\cite{ZhouQ10}. It would be interesting to prove a better lower bound for the capacitated case, if possible. 

As mentioned in Section~\ref{sec:intro}, there is a PTAS known for the uncapacitated rooted $\MMTC$ problem on a tree metric. We conclude with our second open question of proving a PTAS for the hard capacitated version of the same.